\documentclass[conference]{IEEEtran}
\IEEEoverridecommandlockouts

\usepackage[T1]{fontenc}
\usepackage{graphicx}
\usepackage{cite}
\usepackage{amsmath,amssymb,amsfonts}
\usepackage{algorithmic}
\usepackage[caption=false,font=footnotesize]{subfig}
\usepackage{textcomp}
\usepackage{xcolor}
\usepackage{booktabs}
\usepackage{multirow}
\usepackage{amsthm}
\usepackage{tabularx}
\usepackage{enumitem}
\usepackage[most]{tcolorbox}
\usepackage[hyphens]{url}
\usepackage[hidelinks]{hyperref} 
\def\BibTeX{{\rm B\kern-.05em{\sc i\kern-.025em b}\kern-.08em
    T\kern-.1667em\lower.7ex\hbox{E}\kern-.125emX}}

\newtheorem{theorem}{Theorem}

\begin{document}

\title{
Selecting Offline Reinforcement Learning Algorithms for Stochastic Network Control
}


 \author{\IEEEauthorblockN{
 Nicolas Helson, Pegah Alizadeh, Anastasios Giovanidis}
 \IEEEauthorblockA{\textit{Ericsson Research, 
 Massy, France} \\
 email: nicolas.helson@gmail.com, \{firstname.lastname\}@ericsson.com}
 }
\maketitle

\begin{abstract}
Offline Reinforcement Learning (RL) is a promising approach for next-generation wireless networks, where online exploration is unsafe and large amounts of operational data can be reused across the model lifecycle. However, the behavior of offline RL algorithms under genuinely stochastic dynamics—inherent to wireless systems due to fading, noise, and traffic mobility—remains insufficiently understood. We address this gap by evaluating Bellman-based (Conservative Q-Learning), sequence-based (Decision Transformers), and hybrid (Critic-Guided Decision Transformers) offline RL methods in an open-access stochastic telecom environment (\texttt{mobile-env}). Our results show that Conservative Q-Learning consistently produces more robust policies across different sources of stochasticity, making it a reliable default choice in lifecycle-driven AI management frameworks. Sequence-based methods remain competitive and can outperform Bellman-based approaches when sufficient high-return trajectories are available.
These findings provide practical guidance for offline RL algorithm selection in AI-driven network control pipelines, such as O-RAN and future 6G functions, where robustness and data availability are key operational constraints.

\end{abstract}


\section{Introduction}
Next generation wireless networks are built with zero-touch capabilities incorporating full automation, especially for parameter tuning, where learning-based control functions are increasingly deployed and updated as part of managed AI lifecycles (e.g. in O-RAN and future 6G architectures).  Reinforcement Learning (RL) is the key towards achieving the envisioned adaptability and resilience. Value-based and policy gradient methods are able to successfully solve high-dimensional decision-making tasks~\cite{mnih2015human,haarnoja2018soft}. 
However, use of online RL \cite{Yang2020infocom} is challenging because it requires direct interaction with the real or a simulated environment in order to sufficiently explore the impact of the available tuning options. Exploration in real environments can be expensive and unsafe, because a suboptimal configuration risks degrading the network performance and severing the offered service. This is why in practice training is done on simulators which mimic reality with several levels of abstraction. Even on simulators, however, exploration can be slow, because of the time needed to evaluate the outcome of a certain parameter tuning in various network Key Performance Indicators (KPIs). Besides, the trained policy when deployed needs further fine-tuning using real data and exploration in order to adapt to real environment  specificities.

On the other hand, wireless network operators are able to collect large amounts of data from the deployed network in cities and rural areas. These data are often summary KPIs collected with some periodicity or when a change in the network configuration occurs. 
These datasets take the form of multi-dimensional time series describing sequences of states (observation KPIs), actions (parameter tuning), and rewards (target KPIs). Such structured trajectories naturally fit the Markov Decision Process (MDP) framework, making them suitable for training an RL agent. To this aim, offline RL \cite{levine2020offline} emerges as a natural fit, as it can leverage these pre-collected datasets to derive new high-performing  policies, without the need for explicit exploration. This perspective is consistent with the AI/ML lifecycle defined by the O-RAN Alliance, where data collection, data preparation, and model training are explicitly separated stages in the AI/ML workflow for network control functions~\cite{oranWG2AIML}.



Naturally, offline RL comes with new challenges that directly impact algorithm selection and deployment decisions. By relying on fixed datasets collected through some \textit{behavioral policy}, it can lead to inaccurate value estimates for actions that are rarely or never observed \cite{levine2020offline}, referred as action \textit{out-of-distribution} (OOD) problem. 
Additionally, in order for any offline RL method to learn a meaningful policy, the dataset should contain sufficient transitions to high reward regions.

There are two important approaches to treat these issues, one based on traditional Bellman-based formulation, the other on sequential methods, each implying different trade-offs for robustness, data-requirements and operational suitability. \textit{Conservative Q-Learning} (CQL) \cite{kumar2020conservative} is a Bellman-based method, which addresses some of these issues by regularizing the value function to reduce overestimation. More recently, \textit{Decision Transformers} (DTs)~\cite{chen2021decision} and their variations such as Critic-Guided DT~\cite{wang2024critic} have emerged as a promising alternative. By framing RL as a conditional sequence modeling problem, DTs bypass explicit bootstrapping and instead predict actions conditioned on the desired future cumulated rewards, or \textit{return-to-go}, leveraging the advances of large-scale transformer models.

The inherent stochasticity of telecommunication environments poses an additional  significant challenge for offline RL, directly affecting the reliability of autonomous decision making in live networks. This can manifest in different components of the MDP, such as in the \textit{initial state distributions}, \textit{state/action transitions} or \textit{reward} function, with varying degrees of intensity and structure. This is called aleatoric uncertainty in \cite{lockwood2022reviewuncertaintydeepreinforcement}, i.e., uncertainty inherent to the environment that cannot be reduced. Another type of uncertainty, epistemic uncertainty, stems from limitations of the model or the training data and can, in principle, be reduced with additional data or improved modeling. 




\textbf{Contributions:} While prior work has compared Decision Transformers (DT) against Bellman-based methods such as CQL~\cite{omori2025should,bhargava2023should}, these either neglect stochasticity~\cite{omori2025should} or evaluate it after training on deterministic data~\cite{bhargava2023should}. An important first contribution of our work is to compare these methods on an environment with naturally incorporated stochasticity. We conduct experiments on a telecommunication-specific framework, \texttt{mobile-env}~\cite{schneider2022}, which enables controlled stochasticity representative of real-world wireless multi-cell scenarios. The two types of studied stochasticity are: traffic mobility and channel fading.

Most importantly, this study extends the comparison between CQL and DT across different types and levels of aleatoric uncertainty by including sequential methods that have been shown to outperform standard DTs. In particular, we evaluate a DT variant augmented with critic objectives, Critic-Guided DT~\cite{wang2024critic}. For completeness, we also provide a dataset ablation study to assess the impact of epistemic uncertainty, under high stochasticity.

Specifically, our primary question is: \textbf{How does stochasticity affect the performance and deployment of offline RL algorithms for wireless communications?}
To address this, we pose and answer a list of practical sub-questions evaluating the behavior of both Bellman-based and sequence-based offline RL methods:
\begin{itemize}

    \item How is the telecom-agents' performance perturbed by user mobility (state transition stochasticity)? (Section~\ref{subsec:exp-mobile-env-mobility})

    \item How does dataset quality 
    impact performance of the telecom-agents under user mobility?  (Section~\ref{subsec:exp-mobile-env-ablation}) 
    

    \item How is the telecom-agents' performance perturbed by channel fading (reward stochasticity)? (Section~\ref{subsec:exp-mobile-env-fading})  

\end{itemize}




\section{Related Work}
Offline RL 
is increasingly being applied in wireless communication systems. 
Recent work~\cite{10529190} studies the impact of dataset composition on offline RL performance for wireless systems. The study of~\cite{11140155} compares online and offline RL approaches, including DT and CQL, for the downlink adaptation problem. Similarly,~\cite{yang2023advancingranslicingoffline} investigates the use of offline RL—particularly CQL—for RAN slicing optimization. The work in \cite{alizadeh2025handover} learns optimal handover parameter tuning using DT and CQL. However, none of these works explicitly consider stochasticity as a major challenge in offline RL for telecommunications.


The Decision Transformer (DT)~\cite{chen2021decision} represents a prominent class of sequential approaches to offline RL. However, DT faces two main limitations: a weak stitching ability (i.e., difficulty combining suboptimal trajectory segments into improved ones) \cite{wang2024critic,yamagata2023q}, and sensitivity to stochasticity introduced by return conditioning \cite{paster2022you}. It can fail when high returns result from luck rather than consistently optimal behavior. 
%
Several methods have been proposed to address the limitations of DT, typically by incorporating a critic trained using Bellman equation principles. Q-learning DT~\cite{yamagata2023q} adopts CQL-style return relabeling to reduce overestimation bias, while Critic-Guided DT (CGDT)~\cite{wang2024critic} pretrains a critic to guide policy learning, improving the stitching of suboptimal trajectories and mitigating reliance on “lucky” high-return samples. ESPER~\cite{paster2022you} clusters trajectory data to enhance robustness to stochasticity, and ACT~\cite{Gao_Wu_Cao_Kong_Zhang_Yu_2024} introduces advantage conditioning to stabilize training.


There are alternative Bellman-based methods for efficient offline RL. Parallel to Conservative Q-Learning (CQL)~\cite{kumar2020conservative}, introduced earlier, 
Implicit Q-Learning (IQL)~\cite{kostrikov2021offline} mitigates the action OOD problem~\cite{levine2020offline} by restricting value updates to in-distribution actions through an expectile regression objective. 

Offline RL performance under stochasticity has been examined by introducing noise at test time~\cite{bhargava2023should}. Other studies~\cite{paster2022you,Gao_Wu_Cao_Kong_Zhang_Yu_2024} provide limited comparisons between CQL and DT in stochastic settings. The work in~\cite{paster2022you} studies state-transition stochasticity; it is shown that CQL outperforms DT on a low-stochastic Connect Four environment, while DT surpasses CQL on the 2048 game. 
ACT~\cite{Gao_Wu_Cao_Kong_Zhang_Yu_2024} introduces noise into selected actions within D4RL benchmarks, finding that CQL consistently outperforms DT. 


\section{Studied algorithms}
Reinforcement Learning (RL) is grounded in the framework of Markov Decision Processes (MDPs), which consists of states representing the environment, actions that influence state transitions, and rewards that quantify the desirability of outcomes. The state/action-transition probabilities and reward functions are typically unknown or intractable, motivating the data-driven nature of model-free RL. Importantly, MDPs can exhibit different forms of stochasticity. The most common source is the \textit{state transition stochasticity}, where taking a given action in a given state may lead to multiple possible next states according to a probability distribution. Another source is \textit{reward stochasticity}, where the received reward is itself a random variable conditioned on the state and action. While less common in tested environments, in fact reward stochasticity naturally occurs in telecommunication scenarios, principally due to wireless channel fading; as a consequence the received bitrate (reward) is not deterministic. Further sources of randomness include \textit{initial state stochasticity}, because environments are often initiated with a random state; but in telecom environments it is less critical due to the continuous and uninterrupted  operation of the network. 

The objective of RL is to learn a policy that maps each state to an action so as to maximize the expected cumulative future reward, known as the \textit{return}. A related concept often used in sequential methods is the \textit{return-to-go}, defined as the sum of rewards from the current time step $t$ until the end of the trajectory $T$: $\sum_{i=t}^{T} r_i$. 
We next give more detailed insights about the algorithms we study.

\textbf{CQL~\cite{kumar2020conservative}.} 
This follows a classical RL paradigm in which the policy is derived from a learned value function estimating the expected cumulative reward from each state–action pair. To address the action OOD problem during value learning, the authors add a conservative penalty term to the standard Bellman error objective. This penalty discourages assigning high values to unseen actions, while keeping the values of in-distribution actions anchored.

\textbf{DT~\cite{chen2021decision}.}
Concurrently, DT introduces a family of sequence-modeling approaches that do not learn value functions. Instead, it formulates RL as predicting the next action conditioned on the past trajectory. Unlike standard RL policies that rely only on the current state, DT conditions on previous states, previous actions, and also future information. The policy is learned via supervised training of a transformer, which can attend to relevant past elements and key future signals: the return-to-go. At inference time, a desired return target is provided, and the model attempts to produce actions that achieve it. This enables DT to follow different target returns, whereas Bellman-based methods focus on learning a single optimal policy. To allow the model to represent the diverse behaviors in the dataset, a \textit{context length} $K$ is used during both training and inference, restricting how much past information is fed into the model. Larger contexts encourage imitation of dataset trajectories, whereas smaller ones lead to more flexible action predictions.

\textbf{CGDT~\cite{wang2024critic}.} Learning to follow a return goal can be limiting, especially when the dataset contains suboptimal trajectories that include good actions but yield low returns, or in stochastic environments where high returns may result from luck. To address this, CGDT first learns a critic in a supervised manner, which predicts the expected return of an action given the past trajectory context. The DT objective is then modified to balance two components: the vanilla DT supervised action prediction loss conditioned on past states and target returns, and an additional critic-guided term encouraging the model to favor actions predicted to outperform the dataset’s returns. This allows the model to go beyond imitation and aim for returns higher than those observed in the data.

\section{Telecom Environment}



\begin{figure}[t]
    \centering
    \begin{minipage}{0.40\textwidth}
        \centering
        \includegraphics[width=\linewidth]{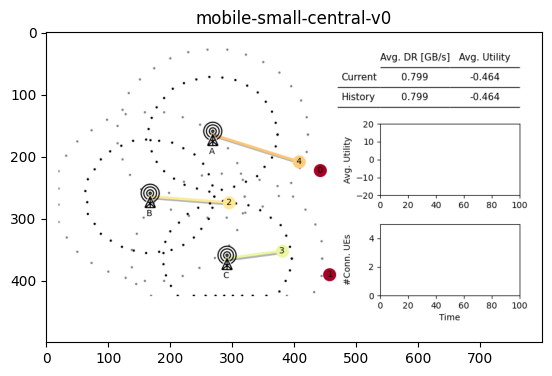}
        \caption{\texttt{Mobile-env} environment illustration.}
        \label{fig:mobile-env}
    \end{minipage}\hfill
    \begin{minipage}{0.40\textwidth}
        \centering
        \includegraphics[width=\linewidth]{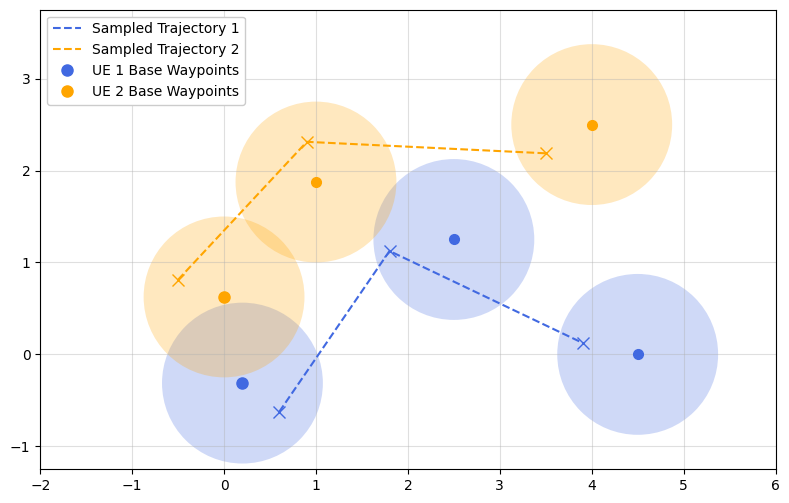}
        \caption{Low-mobility variant of the Random Waypoint model.}
        \label{fig:mov_stochasticity_plots}
    \end{minipage}
\end{figure}

We consider a telecom environment based on \texttt{mobile-env}~\cite{schneider2022}, an open and  lightweight cellular network simulator, designed to allow for user mobility and dynamic user–base station associations
, illustrated in Figure~\ref{fig:mobile-env}. We adapt the original\footnote{Source code: \url{https://github.com/stefanbschneider/mobile-env/tree/main}} environment by replacing its exponentially large action space—which scales with the number of users—with a more realistic formulation using association thresholds, and just three possible actions per base station. This modification is consistent with realistic parameter tuning, while making optimization tractable. The modified action space  version is presented here.


The \texttt{mobile-env} environment consists of multiple base stations (BSs) and user equipments (UEs), where each UE connects to one or more BSs. Users have random mobility trajectories, making the problem challenging. The objective is to maximize the average user utility, defined as a function of the received data rate by the user from the connected station. Intuitively, connecting a user to multiple BSs simultaneously can increase its data rate, but sharing resources among connected users lowers everyone’s rates, creating a load–performance trade-off.


Consider one BS and the users connected to it. A user’s data rate with this BS depends on the Signal-to-Noise-Ratio (SNR), i.e. the signal strength it experiences and on the BS’s scheduling strategy, which governs how resources are shared among connected users. 
In this work, we use the RateFair scheduling strategy, which assigns equal rates to all connected users, meaning that distant users require more resources than nearby ones to achieve the same rate.
The final data rate of a user is then obtained by summing its data rates across all BSs it is connected to.

Data rates depend on user connections with base stations. 
In our version of \texttt{mobile-env} each base station has an SNR threshold controlling which users connect, enabling load regulation by raising or lowering it. In simple words, when a user's SNR exceeds a station's threshold, then the user connects to it. If the threshold is set too high no user connects and reward is $0$. If it is too low, all users connect to all stations and reward is again low because of resource sharing.


We now turn to the RL formulation of \texttt{mobile-env}. We set 5 UEs and 3 BSs to study a small yet sufficiently stochastic and meaningful scenario.


\begin{itemize}
    \item \underline{State Space:} The state includes the current connection thresholds for all base stations (BSs), the current SNR values for each BS-user pair and the users' previous utility. The resulting state dimension is $|B_s| + |U_e||B_s| + |U_e| = 23$ where $|U_e|$ and $|B_s|$ are number of UEs and BSs in the system.
    
    \item \underline{Action Space:} It is discrete. For each base station, the action consists of either slightly increasing, slightly decreasing, or keeping the current threshold. In practice, the update step is set to one-tenth of the total SNR range. The resulting action space thus has $|B_s|$ dimensions and size $3^{|B_s|}=9$.

    \item \underline{Reward:} The reward is defined as the average utility across all UEs. As discussed above, the utility is monotone increasing with the bitrate; the latter depends on the SNR, the user-BS association, and the resources allocated by the BS. Rewards can be stochastic due to channel fading (see below). For more details see the appendix.
\end{itemize}


In this work, we study two sources of stochasticity: 
(i) user mobility, which introduces randomness in the state transition function, and 
(ii) channel fading, which introduces randomness in the reward function. 

\textbf{User Mobility:}
 follows a Random Waypoint (RWP) model~\cite{bettstetter2004rwp}. 
Each user moves at a fixed speed toward randomly sampled waypoints, and new waypoints are drawn once the previous ones are reached. This induces stochastic state transitions since user motion affects future SNR values. To study reduced transition stochasticity, we further define and include in the environment a \emph{low-mobility} variant. Here, each waypoint is drawn within a restricted region around sampled locations, thereby reducing variability in their trajectories. This mechanism is illustrated in Figure~\ref{fig:mov_stochasticity_plots}.


\textbf{Channel Fading:}
is modeled as a random perturbation to the users’ SNR, reflecting effects such as obstacles and multipath propagation. Fading was not originally included in the code, although prevalent in cellular communications, so we included it. We focus on Rayleigh fading, which represents highly stochastic conditions with no line-of-sight component. Fading introduces only reward stochasticity: at each step, a fading coefficient is sampled for every UE–BS pair, and the SNR in the reward computation is scaled by the square of this coefficient. State SNRs remain unperturbed, reflecting the use of averaged measurements in practical systems. Thus, fading affects rewards without altering state transitions.



\section{Experiments}
The first step in offline RL is to collect an offline dataset, otherwise assumed available in real applications. Following common practice in the literature for Gym environments and \texttt{D4rl} dataset ~\cite{fu2020d4rl}, we generated those datasets by training an online RL agent. We used the Double DQN algorithm 
from \texttt{d3rlpy}~\cite{seno2022d3rlpy} for the discrete-actions setting of \texttt{mobile-env}. This procedure yields an expert policy, and by saving model parameters from an earlier stage of training, we obtained a medium performing policy. We then form a medium-expert dataset, containing trajectories sampled using expert policy and trajectories sampled using medium policy. In \texttt{mobile-env}, we built a dataset of 500 expert and 500 medium trajectories of horizon length 100, totaling 100k steps. 
We evaluate trained agents online by running 30 trajectories and reporting the return mean and standard deviation.
 
For readability, we rescale scores to the range $[0,100]$, using the expert’s mean performance as the upper bound and a random policy’s mean performance as the lower bound, following standard offline RL practice~\cite{fu2020d4rl}. Rescaling also enables comparison across different settings (e.g., stochasticity levels).

\subsection{How is the telecom-agents' performance perturbed by user mobility (state transition stochasticity)?}
\label{subsec:exp-mobile-env-mobility}

We now analyze the effect of state-transition stochasticity in \texttt{mobile-env}. 
Using the above described medium–expert datasets and evaluation 
, we train all algorithms under two mobility conditions—limited and high mobility. 
Figure~\ref{fig:data-distribution-mobile-env-mobility} shows the return distributions of the training datasets. 
In the high-mobility setting, the expert and medium returns overlap much more despite similar average gaps, meaning the medium policy can occasionally outperform the expert by chance. 
This amplifies the lucky-return problem and makes the stochastic setting particularly challenging, especially for sequential methods.


\begin{figure}[t]
    \centering
    \subfloat[Limited mobility stochasticity.]{%
        \includegraphics[width=0.48\linewidth]{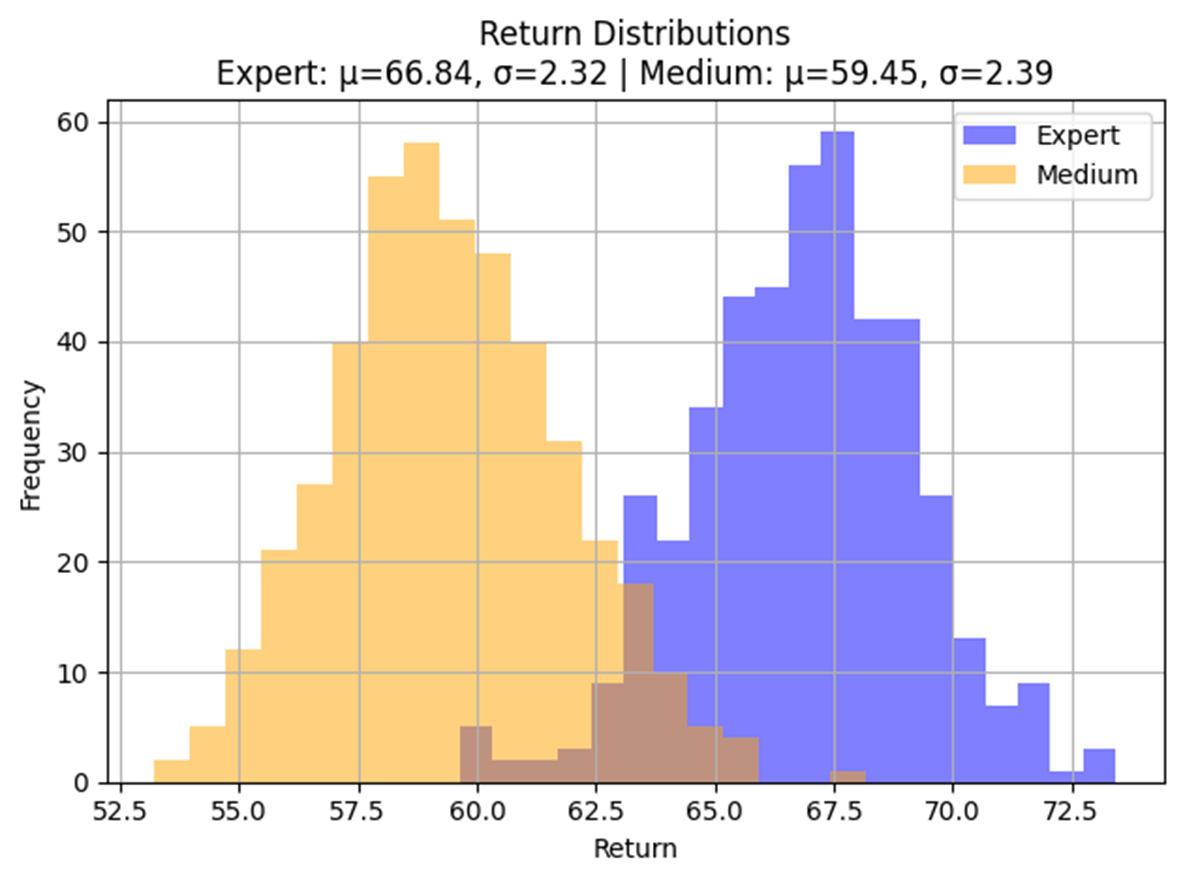}%
        \label{fig:data-dist-lim}
    }
    \hfill
    \subfloat[High mobility stochasticity.]{%
        \includegraphics[width=0.48\linewidth]{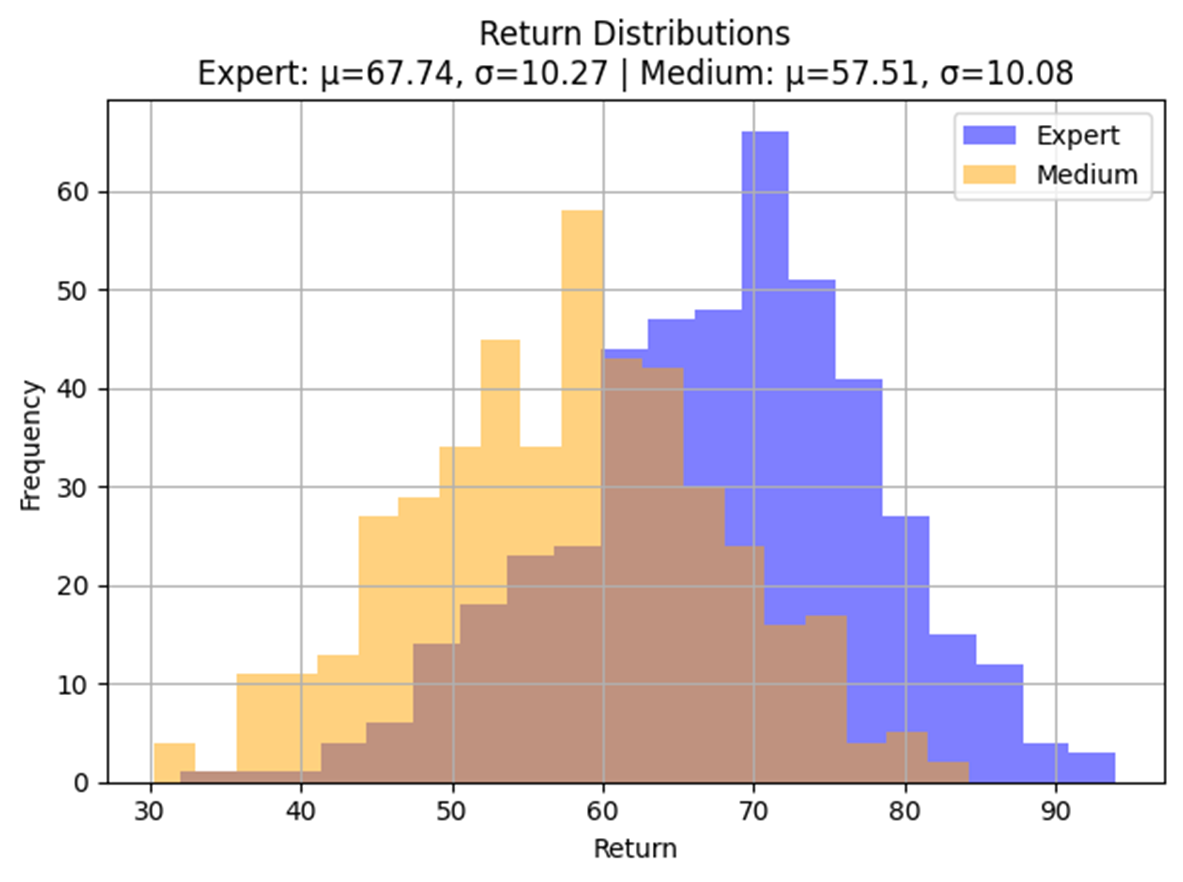}%
        \label{fig:data-dist-high}
    }
    \caption{Return distributions of the training trajectories for \texttt{mobile-env} under different mobility stochasticity settings.}
    \label{fig:data-distribution-mobile-env-mobility}
\end{figure}

\begin{table}[ht!]
\centering
\caption{
        Rescaled mean return (standard deviation in parentheses) under limited and high mobility stochasticity in \texttt{mobile-env}.}
    \begin{minipage}{0.9\linewidth}
        \centering
        \scriptsize
        \begin{tabular}{lccc}
        \toprule
        Model & Limited Stochasticity & High Stochasticity & $\Delta$ (High–Limited) \\
        \midrule
        DT   & 95.1 (9.71) & 81.5 (51.6) & $-13.6$ \\
        CGDT & \textbf{95.6 (9.75)} & 83.0 (51.2) & $-12.6$ \\
        CQL  & 94.6 (13.1) & \textbf{84.8 (46.9)}  & $\mathbf{-9.8}$ \\
        \bottomrule
        \end{tabular}
        \label{tab:results_mob_stoch_comparison}
    \end{minipage}
\end{table}

    


Having characterized the datasets, we now turn to the algorithm results (Table~\ref{tab:results_mob_stoch_comparison}). 
Under limited mobility, all methods achieve similar performance, as expected in a low-stochasticity regime. 
CGDT ($L=5$, $\tau_c=0.5$, $\tau_p=0.8$, $\alpha=0.1$) shows a small but encouraging improvement over DT ($L=5$), while CQL ($\alpha=1.0$, rate $\ell_r=3\cdot 10^{-4}$) exhibits higher variance. 
One explanation is that limited mobility restricts state–action diversity, making it harder for CQL to learn a reliable value function without relying on out-of-distribution actions.

In the high-mobility setting, all methods again reach similar mean performance but with larger variances due to diverse movement patterns. 
Here CQL ($\alpha=0.1$, rate $\ell_r=3\cdot 10^{-4}$) outperforms sequential models, which aligns with its value-based robustness in stochastic environments. 
Increased mobility also exposes more of the state–action space, improving CQL’s value estimates and mitigating OOD issues. 
CGDT ($L=10$, $\tau_c=0.5$, $\tau_p=0.8$, $\alpha=0.2$) continues to show a small but consistent gain over DT ($L=10$). 
Considering the change in performance between limited and high mobility, 
all methods exhibit a substantial drop: DT decreases by 13.6, CGDT by 12.6, and CQL by 9.8. This confirms that mobility-induced stochasticity adversely affects offline RL performance, while also highlighting CQL’s greater robustness relative to sequential methods.

Another notable observation is that sequential methods, particularly DT, still achieve strong performance compared to CQL. This contrasts with 
findings commonly reported in the literature~\cite{bhargava2023should},~\cite{omori2025should}. 
By processing  longer input sequences, sequential models can exploit past state (SNR) trends to anticipate future values; this is especially useful under straight-line mobility, giving them an advantage over Bellman-based methods that rely only on the current state (SNR). 
The dataset may also contain enough expert trajectories to limit the influence of noisy suboptimal ones. 
To examine these factors and assess epistemic uncertainty, we perform an ablation study by progressively removing medium and expert data.

\begin{tcolorbox}[colback=gray!10, colframe=black!20, boxrule=0.3pt, arc=1pt, left=4pt, right=4pt, top=2pt, bottom=2pt]
\textbf{Key result 1:} Agents are negatively affected by user mobility, as higher state transition stochasticity leads to a performance decrease for all methods. However, CQL remains the most robust, showing the highest absolute performance and the smallest performance drop. 
\end{tcolorbox}



\subsection{How does dataset quality impact performance of the telecom-agents under user mobility? (epistemic uncertainty)}
\label{subsec:exp-mobile-env-ablation}

To investigate how data quantity and quality affect the performance of offline RL algorithms under high mobility stochasticity in \texttt{mobile-env}, we conducted an ablation study by progressively removing portions of data from the training sets. We thus increased epistemic uncertainty, under high aleatoric uncertainty.
We used the initially  collected dataset and methodology from the beginning of the section 
, and evaluated three variations: removing 50\% of expert data (\textit{--50\% expert}), 50\% of medium data (\textit{--50\% medium}), and 50\% of both (\textit{--50\% both}). Results are shown in Table~\ref{tab:ablation-study}.  Results are scaled 
using the expert mean  from the complete unreduced dataset. $\Delta$ values indicate the performance difference with respect to scaled results training on the full dataset.

\begin{table}[h!]
\centering
\caption{Ablation study under the \textbf{high mobility} stochasticity.}

\begin{minipage}{0.9\linewidth}
\centering
\scriptsize
\begin{tabular}{llccc}
\toprule
Setting & Model & Mean Return (Std) & $\Delta$ vs. Full \\
\midrule
Full dataset        & DT   & 81.5 (51.6) & - \\
                    & CGDT & 83.0 (51.2) & - \\
                    & \textbf{CQL}  & \textbf{84.8 (46.9)}  & - \\
\hline
-50\% medium/expert & DT   & 73.2 (53.3) & $-8.3$ \\
                    & CGDT & 78.1 (55.6) & $-4.9$ \\
                    & \textbf{CQL}  & \textbf{81.0 (51.2)} & $\mathbf{-3.8}$ \\
\hline
-50\% expert        & DT   & 59.9 (52.8) & $-21.6$ \\
                    & CGDT & 67.9 (52.6) & $-15.1$ \\
                    & \textbf{CQL}  & \textbf{81.4 (48.3)}  & $\mathbf{-3.40}$ \\
\hline
-50\% medium        & DT   & 81.7 (50.5) & $+0.20$ \\
                    & \textbf{CGDT} & \textbf{83.5 (53.4)} & $\mathbf{+0.50}$ \\
                    & CQL  & 82.3 (52.5) & $-2.50$ \\
\bottomrule
\vspace{1pt}
\end{tabular}
\label{tab:ablation-study}
\end{minipage}

\end{table}

First, CQL remains highly stable: it performs well across all settings, demonstrating robustness to stochasticity in various epistemic uncertainty settings, particularly when expert trajectories are limited. Sequential methods, show less robustness. They struggle when expert data is reduced (\textit{--50\% expert} and \textit{--50\% both}), but perform well when only medium data is removed. CGDT consistently outperforms DT across all settings. Overall, CQL achieves the best performance, except in the \textit{--50\% medium} case, where CGDT slightly surpasses it.

In fact, CQL is more sensitive to dataset size than to dataset composition (i.e., the balance between expert and medium trajectories). It achieves its highest performance with the full dataset and its lowest with \textit{--50\% both}, while both \textit{--50\% medium} and \textit{--50\% expert} reductions lead to moderate declines.

For sequential methods, data quality seems to matter more than data quantity. Interestingly, removing 50\% of the medium data does not harm DT or CGDT performance and even slightly improves it. Moreover, when both expert and medium data are reduced (\textit{--50\% both}), these methods perform better than when only expert data is reduced (\textit{--50\% expert}). Too few expert examples may increase the risk of return-conditioning overfitting to “lucky” medium trajectories. On the other hand, having a sufficient proportion of expert data stabilizes learning. This aligns with prior work~\cite{paster2022you} showing that sequential methods can be misled by high-return but unrepresentative trajectories, and suggests a practical mitigation: include enough expert data to limit this bias.

\begin{tcolorbox}[colback=gray!10, colframe=black!20, boxrule=0.3pt, arc=1pt, left=4pt, right=4pt, top=2pt, bottom=2pt]
\textbf{Key result 2:} 
In the high-mobility stochastic setting, CQL outperforms sequential methods across all epistemic-uncertainty configurations, except when medium data is reduced, where CGDT leads. CQL is more sensitive to data quantity, while sequential methods are more affected by data quality.
\end{tcolorbox}

\subsection{How is the telecom-agents' performance perturbed by channel fading (reward stochasticity)?}
\label{subsec:exp-mobile-env-fading}

To examine the impact of reward-level stochasticity, we next study Rayleigh channel fading 
on top of the high-mobility setting, since mobility alone did not reveal substantial performance differences between algorithms and is intrinsically part of the \texttt{mobile-env} framework. As before, we use the dataset sampling, evaluation procedure, and scaling described in the beginning of this  Section. 


\begin{table}[h!]
\centering
\caption{Performance under high mobility with Rayleigh fading.} 
\begin{minipage}{0.9\linewidth}
\centering
\scriptsize
\begin{tabular}{lcccc}
\toprule
Setting & Model & Mean Return (Std) & $\Delta$ vs. no fading\\
\midrule
           & DT   & 47.71  (65.57) & -33.8\\
Rayleigh   & CGDT & 80.38  (62.14) & -2.62 \\
           & \textbf{CQL}  & \textbf{93.97 (64.35)} & \textbf{+9.17} \\
\bottomrule
\end{tabular}
\label{tab:fading-results}
\end{minipage}
\end{table}

Table~\ref{tab:fading-results} reports the rescaled performance of the offline RL algorithms under Rayleigh fading, along with their differences relative to the no-fading case. Scaling is done using mean expert performance from high mobility and Rayleigh fading dataset. $\Delta$ values indicate gap with scaled values (with expert mean from high mobility dataset) from no fading case. CQL remains the most consistent method, achieving the highest mean return and lowest variance, indicating a relative insensitivity to reward stochasticity and confirming the robustness of value-based learning under fluctuating rewards. Sequential methods behave differently: DT performs poorly, likely because strong reward randomness obscures the distinction between good and bad actions, while CGDT offers a substantial improvement over DT but still falls short of CQL.

The deltas relative to the no-fading case provide additional insights. CQL maintains its performance under both state-transition and reward stochasticity. In contrast, DT’s performance drops sharply compared to its no-fading results, suggesting that DT tolerates moderate stochasticity but breaks down when uncertainty increases and stems from multiple sources. CGDT, however, shows much greater robustness, with only minor degradation.

\begin{tcolorbox}[colback=gray!10, colframe=black!20, boxrule=0.3pt, arc=1pt, left=4pt, right=4pt, top=2pt, bottom=2pt]
\textbf{Key result 3:} CQL remains the most consistent method under the addition of Rayleigh fading, while CGDT provides a substantial improvement over DT, maintaining similar normalized performance to the no-fading case.
\end{tcolorbox}

\section{Conclusion}
In this work, we compared offline RL algorithms under different levels of state-transition and reward stochasticity in a realistic telecommunications environment (\texttt{mobile-env}). We evaluated a Bellman-based method (CQL) and sequential approaches (DT and its critic-augmented variant, CGDT), focusing on stochasticity induced by user mobility and channel fading. Our results show that CQL is the most robust overall, while sequential methods remain competitive. From an operational and lifecycle perspective, CQL should be preferred in highly stochastic settings—especially when multiple sources of randomness are present and sufficient data is available. Sequential methods, particularly CGDT, which consistently improves over DT, are better suited for scenarios with milder stochasticity and smaller datasets with high-return trajectories, and may become attractive as data quality improves over the course of system operation.

\bibliographystyle{IEEEtran}
\bibliography{references}

@INPROCEEDINGS{Yang2020infocom,
  author={Yang, Kun and Shen, Cong and Liu, Tie},
  booktitle={IEEE INFOCOM 2020 - IEEE Conf. on Computer Communications Workshops (INFOCOM WKSHPS)}, 
  title={Deep Reinforcement Learning based Wireless Network Optimization: A Comparative Study}, 
  year={2020},
  volume={},
  number={},
  pages={1248-1253},
}

@article{alizadeh2025handover,
  title={Offline reinforcement learning for mobility robustness optimization},
  author={Alizadeh, Pegah and Giovanidis, Anastasios and Ramachandra, Pradeepa and Koutsoukis, Vasileios and Arouk, Osama},
  journal={arxiv preprint arXiv:2506.22793},
  year={2025},
}

@article{mnih2015human,
  title={Human-level control through deep reinforcement learning},
  author={Mnih, Volodymyr and Kavukcuoglu, Koray and Silver, David and Rusu, Andrei A and Veness, Joel and Bellemare, Marc G and Graves, Alex and Riedmiller, Martin and Fidjeland, Andreas K and Ostrovski, Georg and others},
  journal={nature},
  volume={518},
  number={7540},
  pages={529--533},
  year={2015},
  publisher={Nature Publishing Group}
}

@article{levine2020offline,
  title={Offline reinforcement learning: Tutorial, review, and perspectives on open problems},
  author={Levine, Sergey and Kumar, Aviral and Tucker, George and Fu, Justin},
  journal={arXiv preprint arXiv:2005.01643},
  year={2020}
}

@INPROCEEDINGS{schneider2022,
  author={Schneider, Stefan and Werner, Stefan and Khalili, Ramin and Hecker, Artur and Karl, Holger},
  booktitle={NOMS 2022-2022 IEEE/IFIP Network Operations and Management Symposium}, 
  title={mobile-env: An Open Platform for Reinforcement Learning in Wireless Mobile Networks}, 
  year={2022},
  volume={},
  number={},
  pages={},
}

@article{kumar2020conservative,
  title={Conservative q-learning for offline reinforcement learning},
  author={Kumar, Aviral and Zhou, Aurick and Tucker, George and Levine, Sergey},
  journal={Advances in neural information processing systems},
  volume={33},
  pages={1179--1191},
  year={2020}
}

@article{kostrikov2021offline,
  title={Offline reinforcement learning with implicit q-learning},
  author={Kostrikov, Ilya and Nair, Ashvin and Levine, Sergey},
  journal={arXiv preprint arXiv:2110.06169},
  year={2021}
}

@inproceedings{haarnoja2018soft,
  title={Soft actor-critic: Off-policy maximum entropy deep reinforcement learning with a stochastic actor},
  author={Haarnoja, Tuomas and Zhou, Aurick and Abbeel, Pieter and Levine, Sergey},
  booktitle={International conference on machine learning},
  pages={1861--1870},
  year={2018},
  organization={Pmlr}
}

@article{chen2021decision,
  title={Decision transformer: Reinforcement learning via sequence modeling},
  author={Chen, Lili and Lu, Kevin and Rajeswaran, Aravind and Lee, Kimin and Grover, Aditya and Laskin, Misha and Abbeel, Pieter and Srinivas, Aravind and Mordatch, Igor},
  journal={Advances in neural information processing systems},
  volume={34},
  pages={},
  year={2021}
}

@inproceedings{yamagata2023q,
  title={Q-learning decision transformer: Leveraging dynamic programming for conditional sequence modelling in offline rl},
  author={Yamagata, Taku and Khalil, Ahmed and Santos-Rodriguez, Raul},
  booktitle={International Conference on Machine Learning},
  pages={38989--39007},
  year={2023},
  organization={PMLR}
}

@inproceedings{wang2024critic,
  title={Critic-guided decision transformer for offline reinforcement learning},
  author={Wang, Yuanfu and Yang, Chao and Wen, Ying and Liu, Yu and Qiao, Yu},
  booktitle={Proceedings of the AAAI Conference on Artificial Intelligence},
  volume={38},
  number={14},
  pages={15706--15714},
  year={2024}
}

@article{paster2022you,
  title={You can’t count on luck: Why decision transformers and rvs fail in stochastic environments},
  author={Paster, Keiran and McIlraith, Sheila and Ba, Jimmy},
  journal={Advances in neural information processing systems},
  volume={35},
  pages={38966--38979},
  year={2022}
}

@article{bhargava2023should,
  title={When should we prefer decision transformers for offline reinforcement learning?},
  author={Bhargava, Prajjwal and Chitnis, Rohan and Geramifard, Alborz and Sodhani, Shagun and Zhang, Amy},
  journal={International Conference on Learning Representations (ICLR)},
  year={2024}
}

@article{fu2020d4rl,
  title={D4rl: Datasets for deep data-driven reinforcement learning},
  author={Fu, Justin and Kumar, Aviral and Nachum, Ofir and Tucker, George and Levine, Sergey},
  journal={arXiv preprint arXiv:2004.07219},
  year={2020}
}

@article{seno2022d3rlpy,
  title={d3rlpy: An offline deep reinforcement learning library},
  author={Seno, Takuma and Imai, Michita},
  journal={Journal of Machine Learning Research},
  volume={23},
  number={315},
  pages={1--20},
  year={2022}
}

@article{omori2025should,
  title={Should We Ever Prefer Decision Transformer for Offline Reinforcement Learning?},
  author={Omori, Yumi and Dong, Zixuan and Ross, Keith},
  journal={arXiv preprint arXiv:2507.10174},
  year={2025}
}

@article{bettstetter2004rwp,
  title     = {Stochastic properties of the random waypoint mobility model},
  author    = {Bettstetter, Christian and Hartenstein, Holger and P{\'e}rez-Costa, Xavier},
  journal   = {Wireless Networks},
  volume    = {10},
  number    = {5},
  pages     = {555--567},
  year      = {2004},
  publisher = {Springer},
}

@article{Gao_Wu_Cao_Kong_Zhang_Yu_2024, title={ACT: Empowering Decision Transformer with Dynamic Programming via Advantage Conditioning}, volume={38},  abstractNote={Decision Transformer (DT), which employs expressive sequence modeling techniques to perform action generation, has emerged as a promising approach to offline policy optimization. However, DT generates actions conditioned on a desired future return, which is known to bear some weaknesses such as the susceptibility to environmental stochasticity. To overcome DT’s weaknesses, we propose to empower DT with dynamic programming. Our method comprises three steps. First, we employ in-sample value iteration to obtain approximated value functions, which involves dynamic programming over the MDP structure. Second, we evaluate action quality in context with estimated advantages. We introduce two types of advantage estimators, IAE and GAE, which are suitable for different tasks. Third, we train an Advantage-Conditioned Transformer (ACT) to generate actions conditioned on the estimated advantages. Finally, during testing, ACT generates actions conditioned on a desired advantage. Our evaluation results validate that, by leveraging the power of dynamic programming, ACT demonstrates effective trajectory stitching and robust action generation in spite of the environmental stochasticity, outperforming baseline methods across various benchmarks. Additionally, we conduct an in-depth analysis of ACT’s various design choices through ablation studies. Our code is available at https://github.com/LAMDA-RL/ACT.}, number={11}, journal={Proceedings of the AAAI Conference on Artificial Intelligence}, author={Gao, Chen-Xiao and Wu, Chenyang and Cao, Mingjun and Kong, Rui and Zhang, Zongzhang and Yu, Yang}, year={2024}, month={Mar.}, pages={12127-12135} }

@misc{yang2023advancingranslicingoffline,
      title={Advancing RAN Slicing with Offline Reinforcement Learning}, 
      author={Kun Yang and Shu-ping Yeh and Menglei Zhang and Jerry Sydir and Jing Yang and Cong Shen},
      year={2023},
      eprint={2312.10547},
      archivePrefix={arXiv},
      primaryClass={cs.IT},
      url={https://arxiv.org/abs/2312.10547}, 
}

@INPROCEEDINGS{11140155,
  author={Peri, Samuele and Russo, Alessio and Fodor, Gabor and Soldati, Pablo},
  booktitle={2025 IEEE International Conference on Machine Learning for Communication and Networking (ICMLCN)}, 
  title={Offline Reinforcement Learning and Sequence Modeling for Downlink Link Adaptation}, 
  year={2025},
  volume={},
  number={},
  pages={},
}

@ARTICLE{10529190,
  author={Yang, Kun and Shi, Chengshuai and Shen, Cong and Yang, Jing and Yeh, Shu-Ping and Sydir, Jaroslaw J.},
  journal={IEEE Transactions on Wireless Communications}, 
  title={Offline Reinforcement Learning for Wireless Network Optimization With Mixture Datasets}, 
  year={2024},
  volume={23},
  number={10},
  pages={12703-12716},
}

@misc{lockwood2022reviewuncertaintydeepreinforcement,
      title={A Review of Uncertainty for Deep Reinforcement Learning}, 
      author={Owen Lockwood and Mei Si},
      year={2022},
      eprint={2208.09052},
      archivePrefix={arXiv},
      primaryClass={cs.LG},
      url={https://arxiv.org/abs/2208.09052}, 
}

@techreport{oranWG2AIML,
  title        = {O-RAN Working Group 2 AI/ML Workflow Description and Requirements},
  institution  = {O-RAN Alliance},
  number       = {O-RAN.WG2.AIML-v01.03},
  year         = {2023}
}

\section*{Appendix}
\section{Additional experiments on \texttt{LunarLander}}
\subsection{Environment and Stochasticity Setup}
As a preliminary step toward understanding stochasticity-related challenges in Offline RL, we first conducted experiments in the well-known \texttt{LunarLander} environment\footnote{\url{https://github.com/Farama-Foundation/Gymnasium/blob/main/gymnasium/envs/box2d/lunar_lander.py}} before turning to the telecom-specific \texttt{mobile-env}. We present this next.
The goal is to land a spacecraft within a designated target area on the ground. 
The main challenge lies in stabilizing the spacecraft and gradually approaching the landing pad for a safe touchdown. 
The problem becomes stochastic when an external wind is introduced, as it perturbs the spacecraft’s dynamics and affects state transitions.

The \texttt{LunarLander} environment without wind is defined as follows:
\begin{itemize}
    \item \textbf{State space:} The state vector includes the spacecraft’s position, velocity, angular position, angular velocity, and two boolean variables indicating whether the left or right leg is in contact with the ground. Each initial state is randomly sampled: velocity is randomized, while position, angle, and angular velocity are fixed.  
    \item \textbf{Action space:} The action vector has two continuous components $(a_0, a_1)$, where $a_0 \in [-1, 1]$ controls the main engine (vertical thrust), and $a_1 \in [-1, 1]$ controls the lateral engines. 
    \item \textbf{Reward function:} The reward encourages the spacecraft to move closer to the landing pad and make ground contact with its legs. A bonus of $+100$ is awarded for a successful landing, while penalties are applied for excessive speed, large tilt angles, and crashes ($-100$).
\end{itemize}

\textbf{Transition stochasticity (wind):} The spacecraft is perturbed by a stochastic wind signal of the form
\[
w(t) = \tanh\!\big(\sin(2k(t+C)) + \sin(\pi k(t+C))\big),
\]
where $t$ denotes the temporal (sequence) index, $k = 0.01$, $C \sim \text{Uniform}(-9999, 9999)$, and $w(t) \in [-1, 1]$. 
For each trajectory, two independent realizations of $C$ are sampled to generate $w_F(t)$ and $w_T(t)$, which respectively define the lateral wind force and torque applied to the spacecraft. 
The signals are scaled by $\alpha \in [0, 2]$ to control wind strength; we use a relatively strong wind with $\alpha = 1.5$. 
In summary, the spacecraft is perturbed by a lateral force $\alpha w_F(t)$ and a torque $\alpha w_T(t)$, introducing state-transition stochasticity into the dynamics.

\subsection{Experiments}
\label{subsec:exp-lunar-lander}
With \texttt{LunarLander}, our goal is to assess if stochasticity can indeed perturb offline RL algorithms, especially wind-induced state-transition stochasticity. As the action space in this case is continuous, we used SAC\footnote{T. Haarnoja, A. Zhou, P. Abbeel, S. Levine. Soft Actor-Critic: Off-Policy Maximum Entropy Deep Reinforcement Learning with a Stochastic Actor. Proceedings of the 35th International Conference on Machine Learning, PMLR 80:1861-1870, 2018.
\url{https://arxiv.org/abs/1801.01290}} to obtain online an expert policy. To assess stochasticity effects, we trained and evaluated the three Offline RL algorithms under study on \texttt{LunarLander}, both with and without wind. We conducted the trainings and evaluations using the medium-expert datasets and the resulting performance metrics are reported in Table~\ref{tab:results_stoch_comparison}.

These results indicate that CQL outperforms the sequential methods in terms of both mean return and standard deviation, across both deterministic and stochastic settings. Notably, it maintains a clear advantage over DT and CGDT in the high-stochasticity scenario, while all methods struggle to reach expert-level performance (100).

More interestingly, the performance drop between deterministic and stochastic settings highlights each algorithm's robustness to stochastic perturbations. CGDT is the least affected ($\Delta = -0.1$), suggesting that the critic component enhances its robustness to state-transition stochasticity compared to DT. This also helps CGDT narrow the performance gap with CQL in both settings. Conversely, DT experiences a substantial degradation ($\Delta = -6.9$), and CQL, despite its strong overall performance, also suffers a noticeable decline ($\Delta = -4.3$).

In summary, stochasticity does affect Offline RL methods: it notably degrades the performance of DT and CQL, whereas CGDT remains stable across deterministic and stochastic settings. Nonetheless, CQL still achieves the highest overall returns.

\begin{table}[t!]
\centering
\caption{\texttt{LunarLander}: Rescaled mean return and (standard deviation). 
}
\scriptsize
\begin{tabular}{lccc}
\toprule
Model & No Stochasticity & High Stochasticity & $\Delta$ (High–No) \\
\midrule
DT   & 90.4 (16.2) & 83.5 (25.5) & $-6.9$ \\
CGDT & 91.7 (14.5) & 91.6 (19.6) & $\mathbf{-0.1}$
 \\
CQL  & \textbf{101.2 (14.2)} & \textbf{96.9 (17.9)} & $-4.3$ \\
\bottomrule
\end{tabular}
\label{tab:results_stoch_comparison}
\end{table}

\section{Modified Environment \texttt{mobile-env}}

\subsection{Limited Mobility}
\label{app:limited-mobility}
In practice, \texttt{mobile-env} defines a $200 \times 200$ square map. For the limited-stochasticity setting, initial positions are sampled within a circle of radius 20 (one-tenth of the map size), and waypoints within a circle of radius 10 (one-twentieth of the map size).

\subsection{Modifications of \texttt{mobile-env}}
\label{app:modifs-mobile-env}
In the paper, we use a modified version of the original \texttt{mobile-env}. We describe here the original RL setup: The key difference lies in the definition of actions. In the original setup, the action is defined at the user level, and does no make use of the per-base station thresholds we defined in our variation. This means for each user, it specifies whether to connect, disconnect, or do nothing. The action per user is an integer in $\left\{0,\ldots, |B_s|\right\}$, with $|B_s|$ the total number of base stations. More precisely, $0$ indicates no change, while $i$ triggers either a connection or a disconnection with base station $i$. Thus, each user can trigger at most one connection or disconnection per timestep, leading to an overall action space of size $|B_s|^{|U_e|+1}$.  

The original version for the state space includes the binary user connection indicators, rather than the per-BS thresholds we involved in the new version. For each user, it contains (i) a binary vector of size $|B_s|$ indicating the current connections, (ii) the SNR values with each base station, and (iii) the user’s previous utility. The total state vector dimension is therefore $|U_e|\times(2|B_s|+1)$. 

The reward computation stays unchanged in both versions.  

The main computational reason why we decided to work with per-BS thresholds for the association, instead of the original per-UE decision is that the original problem quickly becomes intractable, as the action space grows exponentially with the number of users. This makes efficient exploration nearly impossible and the RL problem overly challenging. Our threshold-based formulation resolves this issue by reducing the action space to $3^{|B_s|}$, while also simplifying the state representation by removing explicit connection indicators (since they are now fully determined by the thresholds). Besides, in real systems low-level association decisions work very similarly to the logic we introduce in the modified version, i.e. with per-BS thresholds that nudge the users to connect to one BS or another based on their relative signal quality. 

We also modified the way SNR values are scaled. In the original version, SNRs are normalized per user, i.e., each user’s SNR values are divided by their own maximum SNR. We believe this makes comparisons across users more difficult and may harm the policy’s learning. Since scaling remains important, we instead introduced fixed lower and upper SNR references to rescale values within these bounds. For the upper reference, we used the SNR corresponding to a user located very close to a base station, while for the lower reference, we used the disconnection frontier, below which mobile-env automatically disconnects the user. 

\subsection{Rewards}\label{app:mobile-env-rewards}

Here, we expose clearly how the formulas behind reward computation work in \texttt{mobile-env}.
We denote base stations by the index $i$ and users by the index $j$. 
Each station--user pair has a signal-to-noise ratio (SNR) denoted by $\gamma_{ij}$. 
The maximum achievable data rate for user $j$ at station $i$ is 
\begin{equation}{\label{data_rate_defintion}}   
d_{ij} = b \, \log\!\bigl(1 + \gamma_{ij} \bigr),
\end{equation}

where $b$ is the bandwidth, assumed identical across stations and users. From a scheduling perspective users share time-slots but not frequency-slots (bandwidth).

We chose to use \textit{RateFair} scheduling policy from \texttt{mobile-env}, meaning each connected user receives the same rate from a given base station. 
Formally, the equal share for station $i$ is obtained as the harmonic mean of the data rates of the users connected to this station:
\begin{equation}
\frac{1}{\sum_j \tfrac{c_{ij}}{d_{ij}}},
\end{equation}
where the connection indicator is defined by
\begin{equation}
    c_{ij} = \mathbf{1}_{\{\gamma_{ij} \geq \tau_i\}}, \qquad c_{ij} \in \{0,1\},
\end{equation}

with $\tau_i$ the connection threshold of base station $i$. Thus, only connected users contribute to the denominator, and each receives the same share.  

The per-user aggregated rate across all stations is given by
\[
f_j(\Gamma) = \frac{b}{n_j} \sum_i c_{ij} a_{ij} \log\big(1 + \gamma_{ij}\big),
\]
where $a_{ij}$ denotes the fraction of the maximum achievable data rate $d_{ij}$ allocated to user $j$ from station $i$, 
$n_j = \sum_i c_{ij}$ is the number of stations connected to user $j$, 
and $\Gamma = (\gamma_{ij})_{i,j} \in \mathbb{R}^{|B_s| \times |U_e|}$ is the matrix of SNR values for all station–user pairs, 
with $|B_s|$ and $|U_e|$ denoting the total number of base stations and users, respectively.  

For instance, under the \textit{RateFair} scheduling policy, we have
\[
a_{ij} = \frac{1}{\sum_k c_{ik} \tfrac{d_{ij}}{d_{ik}}}.
\]
The aggregated rates are then mapped into corresponding utility values.


First, a utility function is applied:  
\[
g(d) = \operatorname{clip}\!\left( w_1 \frac{\log(w_2 + d)}{\log(w_3)}, \, l, \, u \right),
\]
with parameters $w_1, w_2, w_3, l, u$. 
Next, utilities are normalized to $(0,1)$ using,  
\[
h(x) = \frac{x - l}{u - l}.
\]  

Finally, in \texttt{mobile-env}, the reward is defined as the average user utility across users:  
\begin{equation}
r = \frac{1}{|U_e|}\sum_j (h \circ g \circ f_j)(\Gamma).    
\end{equation}

\subsection{Channel Fading}
\subsubsection{Additional Types of Channel Fading}
Channel fading is modeled as a random perturbation to the users’ SNR, capturing the effect of environmental factors such as obstacles (e.g., buildings) on signal propagation. In addition to the Rayleigh fading presented earlier, we also consider two Rician fading settings, with $K = 3$ and $K = 10$. Rician fading models scenarios where the signal reaches the receiver through multiple paths while also including a direct LoS component. Rayleigh fading can be viewed as the special case of Rician fading with no LoS path. Moreover, larger Rice factors $K$ correspond to stronger LoS components and therefore reduced randomness.

\subsubsection{Fading Formulas}
\label{app:fading-examples}

Fading impacts the rewards through its effect on $\gamma$, the signal-to-noise ratio (SNR):
$$
\gamma_{\text{new}} = \gamma_{\text{old}} \, |H|^{2},
$$
where $H$ is a random fading coefficient, drawn from a distribution that depends on the fading model under consideration.

In this work, we study two classical fading models: Rayleigh fading and Rician fading, with parameters $K=3$ and $K=10$.

\textit{Rayleigh fading.}  
This model applies in the absence of a Line-of-Sight (LoS) component, where the received signal results from the sum of many scattered multipath components.  
The fading amplitude $H$ then follows a Rayleigh distribution with probability density function (pdf):  
\begin{equation*}
p_{H}(h) = \frac{2h}{\Omega} e^{-h^{2}/\Omega}, \quad h \geq 0,    
\end{equation*}
where $\Omega = \mathbb{E}[H^{2}] = 1$ denotes the average received power.


\textit{Rician fading.}  When a direct LoS path exists in addition to scattered multipath components, the fading amplitude $H$ follows a Rician distribution. Its probability density function (pdf) is given by:  

\begin{equation*}
\scalebox{0.9}{$
   f_{H}(h) = 
   \frac{2(K+1)h}{\Omega} 
\exp\!\left(-K - \frac{(K+1)h^{2}}{\Omega}\right) 
I_{0}\!\left( 2h \sqrt{\tfrac{K(K+1)}{\Omega}} \right), 
$}
\end{equation*}

where $I_{0}(\cdot)$ is the 0th-order modified Bessel function of the first kind. The parameter $K \geq 0$ represents the \emph{Rician $K$-factor}, which quantifies the ratio between the power in the direct LoS component and the power in the scattered components. We consider two representative cases: $K=3$ and $K=10$, with $\Omega = 1$. Higher $K$ values correspond to stronger LOS conditions, whereas $K=0$ reduces to the Rayleigh fading model.

\subsubsection{Fading Formulation and Property}
\label{app:fading}
Fading modifies the reward as follows. 
For each (user, base station) pair $(i,j)$, there is a fading random variable for the channel power: $|H_{ij}|^2$. 
The new SNR coefficients are $\gamma_{ij} \, |H_{ij}|^2,$ where $\gamma$ is current SNR without fading.
We denote the full SNR matrix under fading by $\Gamma \otimes |H|^2,$
where $\otimes$ denotes elementwise multiplication and where $\Gamma = (\gamma_{ij})_{i,j}$ and $|H|^2 = (|H|_{ij})_{i,j}$.  
The resulting return is now a random variable $R$:  
\[
R = \frac{1}{|U_e|}\sum_j (h \circ g \circ f_j)\!\left(\Gamma \otimes |H|^2 \right),
\]  
where $|U_e|$ is the number of users, and where h, g, and f are defined as in previous reward subsection (Section \ref{app:mobile-env-rewards}).
In contrast, the baseline reward without fading is deterministic:  
\[
r = \frac{1}{|U_e|}\sum_j (h \circ g \circ f_j)(\Gamma).
\]  

We next prove that the expected reward under fading is always less than or equal to the baseline reward, for a general class of scheduling policy where the data rate allocation fraction $a_{ij}$ doesn't depend on the SNRs $\Gamma$.

\begin{theorem}
The expected reward under fading is always less than or equal to the baseline reward:  
\[
\mathbb{E}[R] \leq \mathbb{E}[r] = r.
\]  
\end{theorem}

\begin{proof}
We begin by showing that, for all users $j$, the function $h \circ g \circ f_j$ is concave.  

First, we establish that $f_j$ itself is concave. Recall that $f_j : \mathbb{R}^{|B_s| \times |U_e|} \rightarrow \mathbb{R}$.  
The function $f_j$ can be written as a separable sum over the variables $\gamma_{ij}$:
\[
f_j(\Gamma) = \frac{b}{n_j} \sum_i \phi_{ij}(\gamma_{ij}),
\]
where $\phi_{ij}(\gamma_{ij}) = c_{ij} a_{ij} \log\big(1 + \gamma_{ij}\big)$.  

It is sufficient to show that each $\phi_{ij}$ is concave, which follows directly from the concavity of the logarithm function. Hence, $f_j$ is concave.  
Finally, since $h \circ g$ is a concave and non-decreasing function, the composition $h \circ g \circ f_j$ is also concave.

Since, for all users $j$, the function $h \circ g \circ f_j$ is concave, Jensen’s inequality gives  
\begin{align*}
    \mathbb{E}\!\left[(h \circ g \circ f_j)(\Gamma \otimes |H|^2)\right] 
    &\leq (h \circ g \circ f_j)\!\left(\mathbb{E}[\Gamma \otimes |H|^2]\right) \\
    &= (h \circ g \circ f_j)\!\left(\Gamma \otimes \mathbb{E}[|H|^2]\right) \\
    &= (h \circ g \circ f_j)(\Gamma),
\end{align*}
where the inequality is by Jensen, the first equality uses the fact that $\Gamma$ is deterministic, and the second equality follows from $\mathbb{E}[|H_{ij}|^2] = 1$ for all users $j$ and stations $i$.  

Taking the average over all users yields  
\begin{align*}
\mathbb{E}[R] &= \frac{1}{|U_e|}\sum_j 
\mathbb{E}\!\left[(h \circ g \circ f_j)(\Gamma \otimes |H|^2)\right]\\
&\leq \frac{1}{|U_e|}\sum_j (h \circ g \circ f_j)(\Gamma) = r.
\end{align*}
\end{proof}

\subsection{Additional Fading Results}

To illustrate the previous property, we report the expert dataset mean return under different fading conditions. Table~\ref{tab:fading-expert-results} shows how the expert performance varies with the fading model. 

\begin{table}[h!]
\centering
\caption{Expert dataset under different fading conditions.}
\scriptsize
\begin{tabular}{lcccc}
\toprule
 & Rayleigh & Rician $K=3$ & Rician $K=10$ & No fading \\
\midrule
Expert Mean & 49.01 & 57.79 & 62.64 & 67.74 \\
\bottomrule
\end{tabular}
\label{tab:fading-expert-results}
\end{table}

These values correspond to the upper bounds used for rescaling. As stochasticity increases (from no fading to Rayleigh fading), the expert performance decreases. This aligns with the Jensen-type inequality proved above: the expected reward of an agent under fading is lower than in the no-fading case. The empirical results confirm this trend, with performance consistently decreasing as fading becomes more random.

\begin{table}[h]
\centering
\caption{\texttt{mobile-env}: Performance under the \textbf{high mobility} stochasticity setting with different fading models.
} 
\scriptsize
\begin{tabular}{lcccc}
\toprule
Setting & Model & Mean Return (Std) & $\Delta$ vs. no fading\\
\midrule
           & DT   & 47.71  (65.57) & -33.8\\
Rayleigh   & CGDT & 80.38  (62.14) & -2.62 \\
           & \textbf{CQL}  & \textbf{93.97 (64.35)} & \textbf{+9.17} \\
\hline
             & DT   & 84.18 (56.06) & +2.68 \\
Rician $K=3$ & CGDT & 82.35 (55.88) & -0.65 \\
             & \textbf{CQL}  & \textbf{89.29 (55.29)} &  +4.49\\
\hline
              & DT   & 83.60 (53.71) & +2.10 \\
Rician $K=10$ & CGDT & 77.52 (56.89) & -5.48 \\
              & \textbf{CQL}  & \textbf{92.58 (51.31)} & +7.78\\
\bottomrule
\end{tabular}
\label{tab:fading-results-complete}
\end{table}

We now analyze the additional results under Rician fading ($K=3$ and $K=10$), as shown in Table~\ref{tab:fading-results-complete}. 
CQL remains the most consistent method, achieving the highest mean and lowest variance of returns across all settings. Sequential methods also perform reasonably well. Under Rician fading, DT is less affected and performs comparably well, suggesting that return conditioning can still effectively guide action selection when perturbations are moderate. Interestingly, however, CGDT performs worse than DT in these milder conditions. This counterintuitive result may stem from CGDT’s higher tuning complexity—requiring careful adjustment of the asymmetry parameter, critic weight, and “beating dataset return” parameter during training—which makes it more sensitive to hyperparameter configurations than DT or CQL. 


Note here that, the hyperparameter search related to the CGDT method for this type of fading was not exhaustive, although we tested a number of combinations for the four hyperparameters $(L, \tau_c, \tau_p, \alpha)$. This observation also highlights that as stochasticity increases, hyperparameter tuning becomes more demanding. For methods like CGDT, which involve multiple tuning parameters, this can be particularly challenging, whereas DT and CQL are more robust in this regard.

\section{Hyperparameter selection for \texttt{mobile-env}}



Hyperparameters for \texttt{mobile-env} scenario are summarized in Table~\ref{tab:hyperparameters_mobile_env}.

\begin{table}[h!]
\centering
\caption{Summary of hyperparameters across all \texttt{mobile-env} scenarios.}
\begin{tabular}{ll}
\toprule
\textbf{Scenario} & \textbf{Hyperparameters} \\
\midrule
Limited mobility &
\begin{tabular}[t]{@{}l@{}}
DT: $L = 5$ \\
CGDT: $L = 5, \; \tau_c = 0.5, \; \tau_p = 0.8, \; \alpha = 0.1$ \\
CQL: $\alpha = 1.0, \; l_r = 3\times10^{-4}$ \\
QDT: $L = 5$
\end{tabular} \\
\midrule
High mobility &
\begin{tabular}[t]{@{}l@{}}
DT: $L = 10$ \\
CGDT: $L = 10, \; \tau_c = 0.5, \; \tau_p = 0.8, \; \alpha = 0.2$ \\
CQL: $\alpha = 0.1, \; l_r = 3\times10^{-4}$ \\
QDT: $L = 10$
\end{tabular} \\
\midrule
Rayleigh / Rice fading &
Same as high mobility stochasticity \\
\bottomrule
\end{tabular}
\label{tab:hyperparameters_mobile_env}
\end{table}

    


\section{Additional \texttt{mobile-env} Experiments}

\subsection{Return Conditioned Plots}
\label{app:return-conditioned-plots}
\begin{figure*}[ht!]
    \centering
    \subfloat[DT]{
        \includegraphics[width=0.45\linewidth]{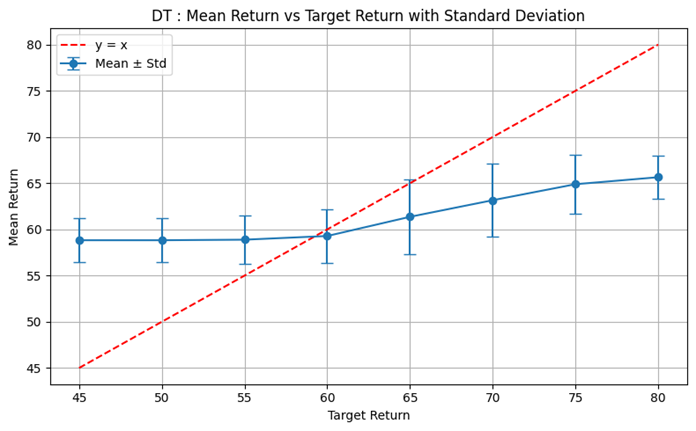}
        }
    \hfill
    \subfloat[CGDT]{
        \includegraphics[width=0.45\linewidth]{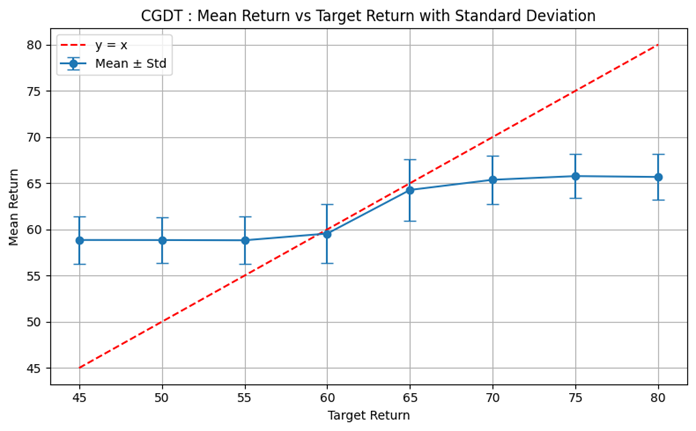}
    }    
    \caption{Target return plots for \texttt{mobile-env} in the limited mobility stochasticity setting with the medium/expert dataset.}
    \label{fig:target_return_plots_limited_stoch}
\end{figure*}

\begin{figure*}[ht!]
    \centering
    \subfloat[DT]{
        \includegraphics[width=0.45\linewidth]{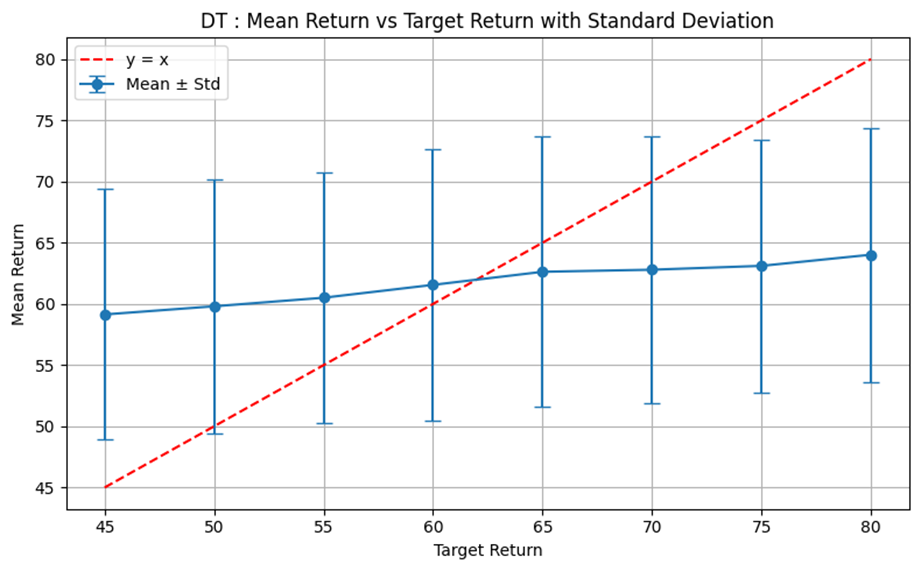}
    }
    \hfill
    \subfloat[CGDT]{
        \includegraphics[width=0.45\linewidth]{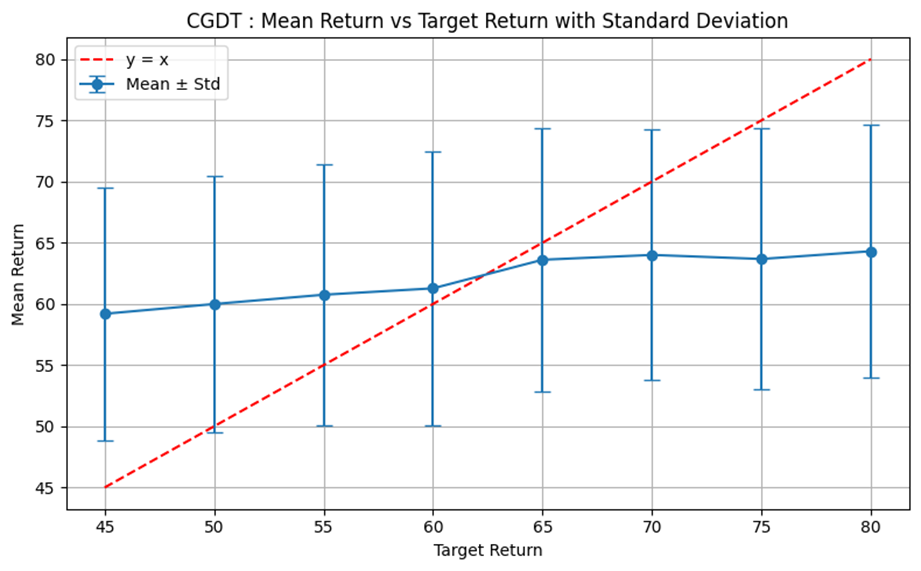}
    }
    \caption{Target return plots for \texttt{mobile-env} in the high mobility stochasticity setting with the medium/expert dataset.}
    \label{fig:target_return_plots_high_stoch}
\end{figure*}

\begin{figure*}[t!]
    \centering
    \subfloat[Limited mobility stochasticity.]{
        \includegraphics[width=0.45\linewidth]{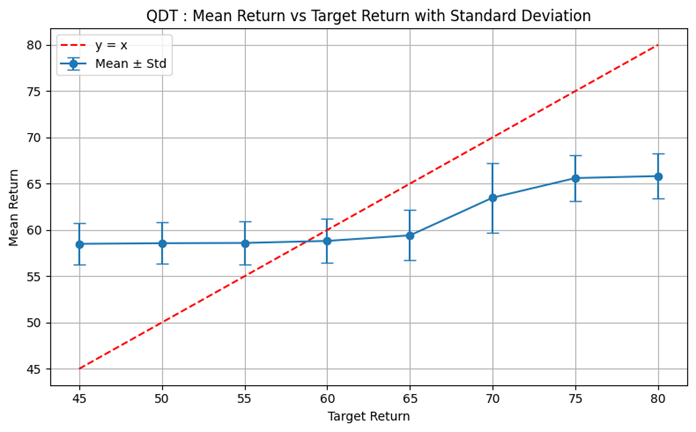}
    }
    \hfill
    \subfloat[High mobility stochasticity.]{
        \includegraphics[width=0.45\linewidth]{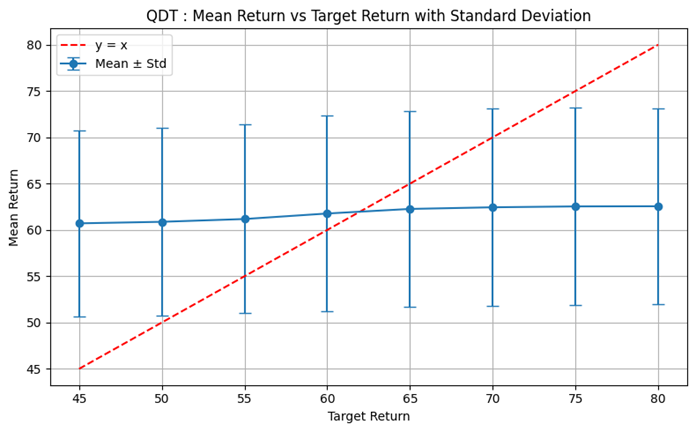}
    }
    \caption{Target return plots of QDT for \texttt{mobile-env} in limited and high mobility stochasticity setting with medium/expert dataset.}
    \label{fig:qdt_target_return_plots}
\end{figure*}

We examined the target return plots for sequential methods, in the limited and high mobility stochasticity settings. These plots display the mean returns obtained when conditioning on different input target returns. Their purpose is to assess the ability of sequential models to follow the specified target and, ultimately, to stitch trajectories in order to achieve the desired performance.

First, let us examine limited mobility case in Figure \ref{fig:target_return_plots_limited_stoch}. The two methods exhibit a similar trend: they begin with a plateau, increase as target returns rise, and then reach a second plateau. This behavior is intuitive, as stitching is possible when sufficient trajectories in the dataset achieve the specified target returns. However, we observe that CGDT tracks the target return more closely than DT. Moreover, for the same target return, CGDT consistently outperforms DT, which aligns with its design goal of achieving returns that match or exceed the observed target return.

We now examine the high-mobility stochasticity case, shown in Figure~\ref{fig:target_return_plots_high_stoch}. 
In this setting, both methods struggle more to follow the input target, confirming that stochasticity perturbs sequential algorithms. 
However, CGDT, aided by its critic, manages to follow the target more closely around the 60-65 return range, where most data are available. 
This again highlights that stochasticity has a stronger negative impact when data are scarce.


\section{Experiments with alternative method: QDT}\label{app:qdt}

\begin{table}[h]
\caption{\textbf{lunar-lander}: Comparison of QDT mean return and (standard deviation) under deterministic and stochastic  environments, against other Offline RL algorithms.}
\centering
\scriptsize
\begin{tabular}{lccc}
\toprule
Model & No Stochasticity & High Stochasticity \\
\midrule
DT   & 90.4 (16.2) & 83.5 (25.5) \\
CGDT & 91.7 (14.5) & 91.6 (19.6)
 \\
CQL  & \textbf{101.2 (14.2)} & \textbf{96.9 (17.9)}  \\
\textcolor{darkgray}{QDT}  & \textcolor{darkgray}{80.0 (26.2)} & \textcolor{darkgray}{80.0 (23.5)} \\
\bottomrule
\end{tabular}
\label{tab:qdt_results_stoch_comparison}
\end{table}

\begin{table}[h]
\caption{\textbf{mobile-env}: Comparison of QDT mean return and (standard deviation) under limited and high mobility stochasticity against other Offline RL algorithms.}
\centering
\scriptsize
\begin{tabular}{lccc}
\toprule
Model & Limited Stochasticity & High Stochasticity \\
\midrule
DT   & 95.1 (9.71) & 81.5 (51.6) \\
CGDT & \textbf{95.6 (9.75)} & 83.0 (51.2)\\
CQL  & 94.6 (13.1) & \textbf{84.8 (46.9)} \\
\textcolor{darkgray}{QDT}  & \textcolor{darkgray}{95.8 (9.88)} & \textcolor{darkgray}{78.3 (49.9)} \\
\bottomrule
\end{tabular}
\label{tab:qdt_results_mob_stoch_comparison}
\end{table}

We also explored an algorithm that combines Bellman-based and sequential Offline RL approaches, called Q-learning Decision Transformer (QDT) \footnote{\url{https://proceedings.mlr.press/v202/yamagata23a.html}}. However, our experiments revealed that it was challenging to achieve satisfactory performance with it. For completeness, we describe the algorithm, report the obtained results, and discuss possible reasons for its struggles.

Q-learning Decision Transformer integrates CQL and Decision Transformer by modifying the trajectory returns. First, a CQL model is fitted on the dataset. Then, a relabeling step is applied using the CQL value function. Starting from the end of each trajectory, returns are recursively relabeled with $R_T = 0$ and
\begin{equation}
    R_{\tau-1} \leftarrow r_{\tau-1} + \max\!\big(R_{\tau}, \hat{V}(s_{\tau})\big),
\end{equation}
where $\hat{V}(\cdot)$ denotes the CQL value function approximation. This relabeling, however, may break consistency with the standard return definition $R_t = r_t + R_{t+1}$. To address this, a second relabeling is introduced: the return-to-go computed in the first step is used as the final return of the input trajectory, while the remaining returns are recomputed using the original dataset rewards. Essentially, this approach modifies only the dataset compared to the standard DT procedure, leaving the rest of the training and inference unchanged. It assumes that the value function is approximated by first running a Bellman-based CQL algorithm before updating the DT.

Tables~\ref{tab:qdt_results_stoch_comparison} and~\ref{tab:qdt_results_mob_stoch_comparison} summarize QDT results on the \texttt{LunarLander} and \texttt{mobile-env} environments, respectively. On \texttt{LunarLander}, QDT performs worse than all other methods under both no stochasticity and high stochasticity. On \texttt{mobile-env}, it outperforms others in low stochasticity settings but lags behind in high stochasticity scenarios. Figure~\ref{fig:qdt_target_return_plots} shows the target return plots for QDT, highlighting its difficulty in following the target under both stochasticity conditions.

During QDT experiments, we observed that the CQL value estimates were often inaccurate. Since QDT relies on these estimates for relabeling, this poses a significant challenge. While CQL captures the relative quality of actions correctly, its absolute value estimates are often unrealistic. Attempts to improve these estimates by increasing CQL training iterations led to instability and worse performance. Given these challenges and the focus of this paper on CQL, DT, and CGDT, we opted to limit the study of QDT.

\end{document}